\documentclass[submission,copyright,creativecommons]{eptcs}
\providecommand{\event}{LTT 2026} 

\usepackage{iftex}

\ifpdf
  \usepackage{underscore}         
  \usepackage[T1]{fontenc}        
\else
  \usepackage{breakurl}           
\fi

\usepackage{amsthm} 
\newtheorem{theorem}{Theorem}[section]
\newtheorem{lemma}[theorem]{Lemma}

\theoremstyle{definition}
\newtheorem{definition}{Definition}[section]

\newtheorem{exmp}[definition]{Example}
\newtheorem{exercise}[definition]{Exercise}

\usepackage{ebproof}
\usepackage{amssymb}
\usepackage{bbold}
\usepackage{mathrsfs}
\usepackage{amsmath}
\newcounter{disgressCounter}
\usepackage[framemethod=tikz]{mdframed}%
\newenvironment{digression}[1][]{\refstepcounter{disgressCounter}
  \begin{mdframed}[
     linewidth=1pt,
     linecolor=blue,
     bottomline=false,
     topline=false,
     rightline=false,
     innerrightmargin=0pt,innertopmargin=0pt,innerbottommargin=0pt,
     innerleftmargin=0.5em,
     skipabove=1\baselineskip
   ]\textbf{\color{blue!75}Digression~\thedisgressCounter. #1} } 
  {\end{mdframed}$\;$}

\newcommand{\eqdef}{=_{\scriptstyle\rm def}}
\newcommand{\powerset}{\mathscr{P}}
\usepackage{todonotes}
\usepackage[normalem]{ulem}

\usepackage{tikz}
\newlength\myheight

\usetikzlibrary{calc}
 
\newcommand{\rtrClo}{{\scriptscriptstyle\Cap}}
\newcommand{\nroClo}{{\scriptscriptstyle\Cup}}
\newcommand{\dedClo}{{\scriptscriptstyle\triangledown}}
\newcommand{\codClo}{{\scriptscriptstyle\vartriangle}}

\title{
Learning Foundations Beneath the Stars
} 
\author{Felice Cardone \qquad\qquad Luca Paolini
\institute{Dipartimento di Informatica\\
University of Turin, Italy}
\email{\{felice.cardone,luca.paolini\}@unito.it}
}
\def\titlerunning{Learning Foundations Beneath the Stars}
\def\authorrunning{Cardone, Paolini}
\begin{document}
\maketitle             
\begin{abstract} 
Foundations of computer science are a key area in theoretical research, one to which
  Stefano  has made significant contributions, particularly from a logical and proof-theoretic perspective. Recently, we have been involved, with him, in teaching an introductory course on this topic, guided by the idea that understanding and writing ordinary, discursive proofs is a valuable skill for future programmers. 
This shared experience has inspired the pedagogical approach at the basis of this paper.

Behind specific foundational topics in computer science lie core techniques that are best taught through examples.
However, standard textbooks often do not place enough emphasis on these ubiquitous techniques and frequently lack examples that are directly relevant to informatics.
We believe that highlighting fundamental techniques, rather than focusing solely on specific foundational topics, would offer significant pedagogical benefits for an introductory course.

In this paper we propose transitive closure of relations as a case study supporting our approach. While all proofs are elementary, we claim that this is a suitable topic for putting to work paradigmatic notions—intrinsically tied to computational thinking—that can serve as structural anchors for a course in the foundations of computer science. In particular, we highlight the techniques employed in proofs, that constitute a comprehensive summary of those that are normally taught in an introductory logic course, and the abstract structures, that allow to connect transitive closure with Kleene star (via quantales) and closure operators (on complete lattices). We then outline a series of further examples that may be used, as in our case study on stars, as a hands-on approach to basic analytic skills to be learned in a course on the foundations of computing.

\end{abstract}

\begin{flushright}
    \footnotesize To Stefano Berardi on the occasion of his birthday.
\end{flushright}

\section{Introduction}
The purpose of this paper is to present one of many possible choices of basic notions that may form the backbone of a course on the Foundations of Computer Science for freshmen (first-year students). Its relatively narrative style is motivated by our focus on the pedagogic relevance of our choice, which leads us to privilege the connections with other topics and techniques rather than on the formal details of the notions that we discuss, which, by the way, are well known.

Our overall proposal is about a change of emphasis over the traditional organization of such a course. While it is rather frequent that courses on foundations of computer science are structured as a series of chapters that deal vertically on fundamental concepts like automata, formal languages, computability and complexity, we believe that the awareness of the methodological coherence of the subject might be improved by developing one or more horizontal projects that may provide cross-sections of the relevant topics within coherent narratives.

These notes develop the essentials of one such project, basically centered on the notion of iteration in two of its key appearances, reflexive transitive closure and Kleene iteration, the stars of our story. In the final section, we suggest other examples that could be used for the same pedagogic purposes as those of the present paper.
\medskip

Before entering our narrative, however, it is worthwhile to spend a few lines to make explicit our attitude towards foundations of computer science as a teaching discipline, especially considering the intended audience of the specific instance of the course that we have in mind.

Throughout our account, we assume that freshmen have already been exposed to the basics of:
 \begin{enumerate}
     \item proof-techniques (direct proof, proof by contradiction and contraposition, induction in its simple form). Here we assume that proof-rules are described as natural deduction rules, in particular as introduction or as elimination rules. We usually employ informally the method of subproofs exploited by Fitch~\cite{fitch1952}, on which there is already a large collection of excellent textbooks addressed to first-year students from many curricula, with many worked-out exercises; 
     \item basic set-theoretic constructions like powerset, Cartesian product, relations, and their properties;
     \item classes of algebraic structures like lattices and monoids, with examples of free structures in these classes, in particular the free monoid over a set. We also assume that the students can recognize easily a collection of elements closed under specific operations as an instance of an algebraic structure: this will be required at least for the powerset, as an instance of the notion of complete lattice or of a complete Boolean algebra, and for the set of words over an alphabet as the free monoid over the alphabet.
 \end{enumerate}
\medskip

According to the mainstream interpretation of ``foundations'', the central topics of the discipline should include at least: formal languages, automata, elements of first-order logic and its models.   From this perspective, foundations may become encyclopedic, 
often encompassing advanced topics that are not suitable for an introductory course. 
We advocate an alternative view, according to which foundations are best understood as a general attitude toward problem solving. This attitude begins with the careful use of natural language as an analytical tool and develops into an appreciation for proofs and their underlying techniques, for the relationships among different notions, and for the advantages gained by exploiting their distinctions in proof construction. Within this framework, our proposal can be seen as the design of a workshop aimed at applying the full potential of this intellectual background to a selection of relevant concrete cases.
 \medskip
 
 The point of departure of our account is the  notion of \textit{transitive closure} of a binary relation over a set, together with several variants of this notion.
Observe that this paper should emphatically \emph{not} be understood as a survey of a technical subject to which many excellent papers already provide 
introductions at many levels of generality (see the last section of this paper for some references). What follows consists of teaching notes for a sequence of lectures that instantiate our general pedagogical approach in a specific case, whose technical details are largely well established.

Our only claim is that the topic chosen is the theme of a coherent project for learning the essential techniques that the course is expected to teach, hoping that this may be useful or even stimulate further discussion on what deserves to be included in a course on foundations like the one we have considered. Often we add digressions to motivate and explain the notions and constructions introduced by framing them within the didactic context that is the leading theme of this paper, and also to sketch their historical background, which in many cases is closely intertwined with the development of the foundations of mathematics.

Our approach can be characterized as \emph{horizontal}, meaning an approach in which themes are chosen carefully so that the proofs do provide the main methodological examples in a significant motivating context. This should be contrasted with a more traditional \emph{vertical} approach where each topic drawn from a fixed repertoire is explored to the desired depth. Of course there are many such examples, and even in this case there exist interesting paths that we have not taken, or are merely suggested and left to the taste of the instructor. An example of these is the brief mention of closure operators and the associated closure systems, that might be used as an explanation of what corresponds to closure in `transitive closure'.

 \section{Transitive closure}
  
 Let $R\subseteq X\times X$ be a binary relation over a set $X$. 
 It may happen that, while $x R y$ and $yRz$ for some $x,y,z\in X$,
 it is not true that $xRz$; namely, the relation $R$ may be not transitive.
However, there are cases where one would like to extend $R$ 
to include all pairs $(x , z)$  where there is $y\in X$ so that $xRy$ and $yRz$.

For example, it might be the case  that  $xRy$ expresses the fact that, starting from a memory state $x$, the execution of one instruction of a program leads to state $y$. In this case, in order to describe the execution of an entire program we are naturally interested in studying sequences
$$x=x_1 R x_2 R \ldots R x_n=y$$
consisting of a finite number (possibly 0) of intermediate steps leading from the input $x$ to the output~$y$.
This new relation, denoted by $R^*$, is generated from $R$ in the sense that it is the smallest reflexive and transitive extension of $R$: the \emph{reflexive transitive closure}  of $R$. 

\begin{digression}\label{dig:closure}
The history of the general notion of closure, of which reflexive transitive closure is an instance, goes back to the origin of abstract algebra and the early work in set theory by Cantor, Dedekind and slightly later by E.\ H.\ Moore. The complete history is reconstructed in Section 1 of the excellent survey \cite{erne:2009}; it is also related to the history of Galois connections and the categorical notion of \emph{adjunction} (and therefore also to the notion of \emph{monad}, that has become important in the study of type systems for programming languages), on which there is another nice survey by the same author \cite{erne:2004}.  
The informal idea at the basis of closure is that of a set $A \subseteq X$ that has to be completed, by adding points, with respect to a property defined over subsets of $X$. It is then natural to regard the completion process as an operator over $X$, namely a function $c:\powerset X \to \powerset X$ defined on subsets of $X$, such that 
\begin{itemize}
\item $c$ is \emph{inflationary}, namely $X \subseteq cX$;
\item $c$ is \emph{idempotent}, $c \circ c = c$;
\item $c$ is \emph{monotonic}, i.\ e.\ $cX \subseteq cY$ whenever $X \subseteq Y$.
\end{itemize}
The first property represents the (possible) addition of points, whereas idempotency corresponds to an interesting linguistic phenomenon that closing shares with other verbs, like filling. A closed set has the form $cX$, and closing a closed set leaves the set unchanged: closed sets are the fixed points of $c$. One may generalize the order-theoretic properties of a closure operator to functions over partially ordered sets, but in this paper we will only need the operator formulation, which is enough to compare reflexive transitive closure to the other star, Kleene's closure. 

Closure operators $c:\powerset X \to \powerset X$ correspond bijectively to \emph{Moore families}, defined as collections $\mathcal{C}$ of subsets of $X$ such that the intersection of any collection of elements of $\mathcal{C}$ is still an element of $\mathcal{C}$. For example, the collection $\tau$ of transitive relations over a set forms a Moore family in this sense (showing that the intersection of a set of transitive relations is a transitive relation is an easy exercise on direct proofs). Transitive closure can therefore be seen as a closure operator associating with every binary relation $R$ over a set the intersection of all elements of $\tau$ that contain $R$. Since the original reference \cite{moore1910} this correspondence has now become folklore.
\end{digression}

\begin{definition}[Reflexive Transitive Closure]
Let $\mathbb{1}$ denote the identity relation on $X$ and,
let $\circ$ denote the composition of relations\footnote{If $R,S$ are relations then $R\circ S \mathrel{:=} 
    \{(x,z)\mid xRy \text{ and } ySz\text{, for some } y\}$, where $\mathrel{:=}$ denotes a definition.}.
 If $R$ is a binary relation on a set $X$ then, its transitive closure 
 $R^*\subseteq X\times X$ is the least relation $S\subseteq X\times X$ such that:
 \begin{itemize}
    \item $R\subseteq S$, \item $\mathbb{1}\subseteq S$, and \item $S\circ S \subseteq S$.
\end{itemize}
\end{definition}

Clearly,  $R^*\subseteq X\times X$ has the following properties:
 $R\subseteq R^*$;
     $\mathbb{1}\subseteq R^*$;
     $R^* \circ R^* \subseteq R^*$.

\begin{exercise}
\begin{itemize}
    \item Show that a binary relation $S$ is transitive exactly when 
    $S\circ S \subseteq S$. Of course use the ordinary definition of transitivity, whereby $S$ is transitive iff $x S y$ whenever $x S z$ and $z S y$, for all $x,y,z$.
    \item Recall that a relation on a set $X$ is said to be dense if, for all $x,y \in X$ for which $x R y$,
     there is a $z\in X$ such that $x R z$ and $z R y$.
    The inclusion\footnote{
    Note that this part of the exercise involves an operative understanding of negation of quantified propositions.} $S \subseteq S\circ S $
    is a characteristic of dense relations: show that there are transitive relations that are not dense. 
\end{itemize}
\end{exercise}

The problem now is how to build $R^*$. There are at least four equivalent ways,
each of them representing a different perspective on the construction of the transitive closure. Our claim is that the proofs of their equivalence involve most of the proof techniques that students should be able to use fluently by the end of the course. In the next section, we will complement these techniques with others relative to complete lattices, and this will justify our choice of (reflexive) transitive closure as a significant case study.
The relation built according to each of these methods will be named $R^\rtrClo, R^\nroClo,R^\dedClo,R^{\codClo}$.
Each of these has a history and a personality of its own.

\begin{definition}\label{def:closures}
\label{def:tc}
Let $R$ be a binary relation on a set $X$.
    \begin{enumerate}
    \item We define $\displaystyle R^\rtrClo := \bigcap_{\mathbb{1}\cup R\cup (S\circ S) \subseteq S} S$.
      \item  We define $\displaystyle R^{\nroClo} := 
      \bigcup_{n\in\omega} R^n$ 
      \quad where $R^0:= \mathbb{1}$ and $R^{n+1}:= R\circ R^n$.
      \label{def:tc-power}
     \item The following rules, where $x,y,z\in X$, define a formal system whose
     theorems identify all and only the pairs belonging to $R^{\dedClo }$:
$$\begin{prooftree}[center=false]
   \hypo{  }
\infer1[\scriptsize{(id)}]{xR^{\dedClo } x   }
\end{prooftree}
\qquad
\begin{prooftree}[center=false]
\hypo{ xRy}
\infer1[\scriptsize{(in)}]{x R^{\dedClo } y }
\end{prooftree}
\qquad
\begin{prooftree}[center=false]
\hypo{xRy }\hypo{yR^{\dedClo } z }
\infer2[\scriptsize{(tx)}]{xR^{\dedClo } z }
\end{prooftree}$$
(The second rule is redundant, see Exercise~\ref{ex:rules} below.)
\item Let $R[A] := \{ b \mid \exists a\in A, aRb\}$, then  we define 
$$R^{\codClo}:= 
\{ (x,y) \mid \forall A\subseteq X \pmb{(} \mbox{ if }( R[A] \subseteq A \mbox{ and } x\in A )\textrm{ then }  y\in A \pmb{)} \}.$$
    \end{enumerate}
\end{definition}

\begin{digression}\label{dig:ancestral}
The definition of $R^\codClo$ is basically the same as the ancestral relation defined originally by Frege 1879, 
and used by Russell and Whitehead in \textit{Principia Mathematica} \cite{russell1910pm}, and then by Quine, in his 1940 textbook \cite{quine1940ml} and later in his treatise on set theory \cite{Quine1963}. In the original treatments of the ancestral the converse $R^{-1}$ is used instead of $R$: transitive closure is, in this sense, dual to the ancestral, the former looking forward to descendants, the latter looking backwards to ancestors. The clause $R[A] \subseteq A$ expresses the fact that $A$ is a \emph{hereditary} subset of $X$: this means that, for any $a \in A$, its $R$-successors are all in $A$, the $R$-successors of the $R$-successors of $a$ are in $A$, and so on, so that, intuitively, the $R$-descendants of any element of $A$ are also in $A$. The whole definition then means that $x$ is related to $y$ in the transitive closure of $R$ whenever $y$ belongs to every hereditary subset of $X$ that contains $x$. 
\end{digression}

    \begin{exercise}\label{ex:rules}
        Show that the second rule defining $R^\dedClo$ is redundant.
    \end{exercise}

   \begin{exmp} 
Let $X=\{ a,b,c \}$ and $R=\{(a,b),(b,c)\}$ a binary relation on $X$.
\begin{itemize}
\item The subsets of $X$ that contain $a$ are four: $A_0=\{a\}$, $A_1=\{a, b\}$, $A_2=\{a, c\}$ and $A_3=\{a,b,c\}$.  We have $R[A_0]=\{b\}$, $R[A_1]=\{b,c\}$, $R[A_2]=\{b\}$
and $R[A_3]=\{b,c\}$ (see Definition~\ref{def:closures}.4).
Since the only $A_i$ such that $R[A_i]\subseteq A_i$ 
is $A_3$, it is easy to see that $\{(a,a),(a,b),(a,c)\} \subseteq R^{\codClo}$.
\item The subsets of $X$ that contain $b$ are four: $B_0=\{b\}$, $B_1=\{a, b\}$, $B_2=\{b, c\}$ and $B_3=\{a,b,c\}$.  We have $R[B_0]=\{c\}$, $R[B_1]=\{b,c\}$, $R[B_2]=\{c\}$
and $R[B_3]=\{b,c\}$.
Since only $B_2$ and $B_3$ are such that $R[B_i]\subseteq B_i$,
it is easy to see that $\{(b,b),(b,c)\} \subseteq R^{\codClo}$.
\item We leave the same reasoning about $c$ as an exercise.
\end{itemize}
    \end{exmp}

\noindent We can now state the main theorem.

 \begin{theorem}\label{th:main}
     $R^*= R^\rtrClo = R^\nroClo = R^\dedClo  = R^\codClo$.
 \end{theorem}

 The proof of the Theorem is subdivided into several lemmas.

\begin{lemma}
    $R^\rtrClo\subseteq R^\nroClo$
\end{lemma}
\begin{proof}
    It is enough to show that $\mathbb{1}\cup R\cup (R^\nroClo\circ R^\nroClo) \subseteq R^\nroClo$.
    Clearly $\mathbb{1} \subseteq R^\nroClo$ and $R \subseteq R^\nroClo$ because 
      both $R^0=\mathbb{1}$ and $R^1=R\circ R^0 = R\circ \mathbb{1} = R$ by Definition~\ref{def:closures}.2.
    Moreover, we also note that:
\begin{itemize}
    \item  $x R^\nroClo y$    implies that there is $n\in\omega$ such that $xR^n y$,
    \item  $y R^\nroClo z$    implies that there is $m\in\omega$ such that $yR^m z$,
\end{itemize}
   so $x R^\nroClo z$ because $xR^{n+m} z$.
\end{proof}

\begin{digression}\label{dig:impredicative}
This first part of the proof follows from the fact that transitive closure 
of $R$ is the \emph{smallest} extension of $R$ which is transitive. 
Here `smallest' is relative to set-theoretic inclusion, whence the (impredicative) 
meet over \emph{all} transitive extensions of $R$ .
Historically, this is important for the criticisms of Poincaré
against the impredicative definitions that are the kernel of
his negative attitude toward the Frege-Whitehead-Russell definition of the ancestral.
Technically, the definition of $R^\rtrClo$ invites a reflection on the idea of closure under a property
and the relation between closure systems and closure operators,
as described, for example, in \cite{davey1990introduction}. We shall come back to this topic in section \ref{further}.
\end{digression}

\begin{lemma}
    $R^\nroClo \subseteq R^\dedClo $
\end{lemma}
\begin{proof}
By induction we prove that, for all $n\in\omega$ we have $R^n\subseteq R^\dedClo $. Clearly $R^0=\mathbb{1}\subseteq R^\dedClo $ by the first rule.
$R^{n+1}=R\circ R^n$,  by induction hypothesis $R^n\subseteq R^\dedClo $, so $R^{n+1}\subseteq R^\dedClo $ by the third rule.
\end{proof}

 Strictly speaking, the following lemma could be deduced from the two preceding lemmas, but it is instructive to present a direct proof.

\begin{lemma}
    $R^\rtrClo \subseteq R^\dedClo $
\end{lemma}
\begin{proof}
From rules (id) and (in) easily follow  $\mathbb{1}  \subseteq S$ and $  R \subseteq S$. It remains to conclude $ S\circ S \subseteq S$.
This follows, because we can prove that the rule
$$\begin{prooftree}
      \hypo{xR^\dedClo  y}
  \hypo{yR^\dedClo z}
  \infer2[\scriptsize(trx)]{xR^\dedClo  z}
\end{prooftree}$$
is admissible, namely it can be added to the formal system without increasing the set of theorems (see Digression \ref{dig:admissibility}).
 In particular, we show that for each derivation of the form
 $$\mathfrak{D}_{trx} = {\begin{prooftree} 
 \hypo{}\ellipsis{$\mathfrak{D}_L$}{\quad xR^\dedClo  y} 
 \hypo{}\ellipsis{$\mathfrak{D}_R$}{\quad yR^\dedClo z}
 \infer2[\scriptsize(trx)]{xR^\dedClo  z}
\end{prooftree}}$$
with 
no occurrence of (trx) in $\mathfrak{D}_L, \mathfrak{D}_R$
we can build a derivation 
 with the same conclusion that does not use (trx), by induction on the height of the subderivation $\mathfrak{D}_L$.\\
There are three cases depending on the last rule applied in $\mathfrak{D}_L$.
\begin{itemize}
\item If the last rule in $\mathfrak{D}_L$ is (id) then the desired derivation is  $\mathfrak{D}_R$, which does not contain applications of (trx) by assumption.\\[-5mm]
\item If the last rule  in $\mathfrak{D}_L$  is (in)  then  
the derivation has the form 
$
\begin{prooftree} 
  \hypo{xRy}
  \infer1[\scriptsize(in)]{xR^\dedClo  y} 
  \hypo{}\ellipsis{$\mathfrak{D}^0_R$}{\quad yR^\dedClo z}
 \infer2[\scriptsize(trx)]{xR^\dedClo  z}
 \end{prooftree}$
 . We can conclude by replacing the proof tree with 
 $\begin{prooftree}[center=false]
  \hypo{xRy}
  \hypo{}\ellipsis{$\mathfrak{D}^0_R$}{\quad yR^\dedClo z}
 \infer2[\scriptsize(tx)]{xR^\dedClo  z}
 \end{prooftree}$.
\item If the last rule   in $\mathfrak{D}_L$ is (tx) then  
the derivation has the form 
  \begin{equation}
      \label{cut}
  \begin{prooftree}
 \hypo{\quad xR u}
 \hypo{}\ellipsis{}{\quad uR^\dedClo y}
 \infer2[\scriptsize(tx)]{xR^\dedClo  y}
 \hypo{}\ellipsis{$\mathfrak{D}^0_R$}{\quad yR^\dedClo  z}
 \infer2[\scriptsize(trx)]{xR^\dedClo  z}.
\end{prooftree}
 \end{equation}
Thus, we can build the derivation\\[-5mm]
$$ \begin{prooftree}
  \hypo{\quad xRu}
      \hypo{}\ellipsis{}{\quad uR^\dedClo  y}
      \hypo{}\ellipsis{$\mathfrak{D}^0_R$}{\quad yR^\dedClo z}
  \infer2[\scriptsize(trx)]{uR^\dedClo  z}
\infer2[\scriptsize(tx)]{xR^\dedClo  z}
\end{prooftree}$$
But the induction hypothesis applies to the derivation of $uR^\dedClo  y$, therefore there exists a proof of
$$ \begin{prooftree}
      \hypo{}\ellipsis{}{\quad uR^\dedClo  y}
      \hypo{}\ellipsis{$\mathfrak{D}^0_R$}{\quad yR^\dedClo z}
  \infer2[\scriptsize(trx)]{uR^\dedClo  z}
\end{prooftree}$$
that does not use (trx), showing that the application of (trx) in \ref{cut} can be eliminated. 
\end{itemize} 
\end{proof}

\begin{digression}\label{dig:admissibility}
Here, we could open a digression on formal systems and the notion of admissible rule,
introduced by Lorenzen in his 1955 treatise on the operative foundations of logic  and mathematics \cite{lorenzen1955springer}. Admissibility is a finer property than derivability, and establishing that a rule in a formal system is admissible normally involves an analysis, within a proof by induction on the height of derivations, of the possible ways in which a conclusion may have been obtained. This makes proofs of admissibility an ideal candidate to get students to become familiar with this flavor of proof by induction.
\end{digression}

\begin{lemma}
    $R^\dedClo \subseteq R^\rtrClo$
\end{lemma}
\begin{proof}
We have to show that, if $xR^\dedClo y$ then $xSy$ for every $S$ such that $\mathbb{1}\cup R \cup (S\circ S)\subseteq S$.
This is proved by induction on the height of the derivation $xR^\dedClo y$ (see Digression \ref{dig:proofTech}).
It is enough to consider the induction step when the last rule in the derivation of $xR^\dedClo y$ is (tx):
$$\begin{prooftree}
 \hypo{xRu}\hypo{uR^\dedClo y}
  \infer2{xR^\dedClo y}
\end{prooftree}$$
The derivation of $uR^\dedClo y$ is shorter than that currently being examined, 
so the induction assumption applies. Thus, $uSy$ for every $S$ satisfying
$$\mathbb{1}\cup R\cup (S\circ S) \subseteq S\;.$$
Moreover $R\subseteq S$ for every such $S$, so we can conclude that $xSu$.
Since $S\circ S \subseteq S$, we can conclude that $R^\dedClo \subseteq R^\rtrClo$. 
\end{proof}

\begin{digression}\label{dig:proofTech}
    Observe that the proof technique by induction on the height of formal derivations in a (simple) formal system is usually employed when proving the properties of $R^\dedClo $.
\end{digression}

\begin{lemma}
    $R^\dedClo \subseteq R^\nroClo$
\end{lemma}
\begin{proof}
The proof is done by induction on the height of the proof $xR^\dedClo y$, and
we show that $xR^ny$ for some $n\in\omega$.
When the derivation is concluded by the first rule, then $x=y$ and $xR^0y$. Otherwise, if $xR^\dedClo y$, because $xRy$, then we have $xR^1y$.
Finally, if $xR^\dedClo y$ follows by the third rule,
by induction hypothesis we have that $yR^nz$ for some $n$, therefore $xR^{n+1}z$.
\end{proof}

Now we know that the first three characterizations of transitive closure of $R$ are equivalent,
and we still need  to show that they also coincide with that of $R^\codClo$.
Consider $R^*$ defined as $R^\nroClo$ above, and for $x\in X$ 
let $\texttt{desc}(x)=\{y\mid x  R^\nroClo y\}$.\bigskip

\begin{lemma}
    $R^\codClo = R^\nroClo$
\end{lemma}
\begin{proof}
Immediately $R[\texttt{desc}(x)]\subseteq \texttt{desc}(x)$ and $x\in\texttt{desc}(x)$,
thus $R^\codClo\subseteq R^\nroClo$.
Conversely, take any $A\subseteq X$ such that $x\in A$ and $R[A]\subseteq A$,
and show by induction on $n$ that $xR^ny$ implies $y\in A$.
The basis is obvious because $x\in A$ by assumption.
By assuming that $xR uR^n y$ we have $x\in A$, $u\in R[A]$ and, in turn, $u\in A$. We have $y\in A$
by the assumption on $A$ and the induction hypothesis. Thus, we conclude that the ancestral of $R$ 
and the transitive closure of $R$ are the same.
\end{proof}

This concludes the proof of Theorem~\ref{th:main}, that gives four equivalent characterizations of (reflexive) transitive closure.

\section{Quantales}

Now we take a slightly more abstract view of the characterization of $R^*$. 
Until now we have been working in the structure 
 $$\powerset(X\times X)$$
 of binary relations over $X$ with the composition of relations as a monoidal operation and arbitrary union and intersection as lattice operations. This structure is an instance of a \emph{quantale} \cite{rosenthal1990lst}.

\begin{definition}[Unital Quantale]
A \emph{quantale} is, basically, a complete\footnote{A complete lattice is a partially ordered set $(L, \leq)$ such that every subset $A$ of $L$ has both a greatest lower bound (aka infimum or meet) $\bigwedge A$ and a least upper bound (aka supremum or join) $\bigvee A$ in $L$. In the special case where A is the empty set, the meet of A is the greatest element of L. Likewise, the join of the empty set is the least element of L. }
lattice $(L,\leq)$ endowed with an internal operation
$\cdot : L \times L \to L$ which is associative and satisfies, for all $X,Y \subseteq L$:
$$
x \cdot (\bigvee Y) = \bigvee \{x \cdot y \mid y \in Y \}
\qquad\mbox{ and }\qquad
(\bigvee X) \cdot x = \bigvee \{y \cdot x \mid y \in X \}.
$$
The quantale is \emph{unital} if it has an identity element 1 for its multiplication:
$$x \cdot 1 = x = 1 \cdot x$$
for all $x \in L$. (In this case, the quantale is naturally a monoid.) 
\end{definition}
Definition~\ref{def:closures}.1 of $R^*$ can be easily written in any unital quantale 
, which is the appropriate structure for generalizing the definition of (reflexive) transitive closure and which deserves a special mention because it plays a unifying role in the topics of our (not so) hypothetical course on foundations.

\medskip\noindent  
Let $R$ be a relation on $X$. c
In a unital quantale we can write $f(S):= 1 \vee (R \cdot S)$ and define:
$$
R^*=\bigwedge_{f(S) \leq S}   S
$$
and we can prove that whenever $f(S)\leq S$ we have $R^*\leq S$;
furthermore, 
\begin{equation}
\label{fix}
f(R^*) = R^*\,,
\end{equation}
 so $R^*$ is the least fixed-point of $f$. 

 \begin{digression}
 The construction of $R^*$ is a special case of the construction of the least fixed point of a monotonic function of a complete lattice into itself, which is (part of) the content of a celebrated theorem by Tarski \cite{tarski:1955}. The proof of the theorem involves only the most basic facts about complete lattices and can therefore be appreciated by the public for which our lectures are designed, so it is a candidate for an extended exercise.
 Another possible -- although slightly more advanced -- topic for meditation starts with the observation that the above definition characterizes $R^*$ as the initial $f$-algebra in a (degenerate) category, and the fixed point property follows then from Lambek's Lemma \cite{AdamekMiliusMoss2025}. Unfolding the categorical terminology, given a category $\mathcal{C}$ and an endofunctor of $F$,
 an \emph{$F$-algebra} is an object $X$ with a morphism $a : FX \to X$ in $\mathcal{C}$. 
 Such an algebra is \emph{initial} if,
 for any other $F$-algebra $b: FY \to Y$ there is exactly one morphism $f: X \to Y$ in $\mathcal{C}$
 such that $b \circ Ff = f \circ a$.
 Lambek's Lemma states that if $X,a$ is an initial $F$-algebra, then $a$ is an isomorphism.
 If we specialize to partial orders, an $F$-algebra is an inequality $FX \leq X$; it is initial, and then $F X = X$, if $X \leq Y$ for any $F$-algebra $Y$. In particular, $R^*$ is the initial $F$-algebra for $F(S):= \mathbb{1} \cup (R \circ S)$.
 \end{digression}
 
 From our present point of view, an important observation is that, by generalizing to quantales, we can introduce the second star of our story by suitably changing the underlying quantale. Consider the quantale of formal languages over an alphabet $A$,
 defined as $$\powerset (A^*)$$
 where $A^*$ is defined as the union $\bigcup_{n\in\omega} A^n$ and where 
 $$ A^0 =\{\varepsilon \}\,,\qquad A^{n+1}=\{a w \mid w\in A^n \ \text{ and }\  a\in A\}. $$     
 Then the monoidal structure is given by 
 $$ LM:= \{ uv \mid u\in L, v\in M \}$$
 for  $L,M\subseteq A^*$.
 
\begin{digression}
    Of course $A^*$ is isomorphic to the set of lists over $A$,
    the initial algebra of the set endofunctor $X\mapsto 1+(A\times X)$, where the empty string $\varepsilon$ corresponds to $\langle \rangle$, and a string of the form $au$ corresponds to the pair $\langle a,u\rangle$.
\end{digression}

The star operation of language theory is then defined by the fixed point equation:
$$ L^* = 1 \cup LL^*\;$$
that provides the value of Equation~(\ref{fix}) in the quantale of languages $\powerset (A^*)$, where $1 = \{\varepsilon \}$.

\begin{digression}
\begin{itemize}
\item The structures introduced above should have made the students familiar with the idea that algebra does not necessarily have to do with numbers. Concerning the algebraic structure of languages, a natural question is whether the \emph{regular} languages are subject to algebraic laws. Various types of structures have been proposed, starting from \cite{conway1971regular}: among these, Kleene algebras are widely studied, but the techniques employed go beyond what may be appreciated by first-year students.

\item One interesting point that might be pursued further is the categorical theme.
     After observing that a quantale is a special case of a complete and cocomplete monoidal category,
     one can define monoid objects in such categories.
     Then, adapting the definition of free object to this impoverished context, the free monoid object over $R$ in the quantale of relations
     turns out to be the reflexive transitive closure of $R$
     whereas the same construction in the quantale of languages gives Kleene's star. In fact, $a^*$ in these quantales satisfies $1 \vee a \vee (a^* a^*) \leq a^* $, which entails the three inequations that state that $a^*$ is a monoid object extending $a$, namely $1 \leq a^*$, $a \leq a^*$ and $a^*a^* \leq a^*$.
     In both cases, freeness of $a^*$ becomes the property that, if $1 \vee a \vee bb \leq b$, then $a^* \leq b$, which follows from  \cite[Theorem 1]{pratt:1991}.
\end{itemize}
\end{digression}

 \section{Algorithmics}
 The logical and algebraic slant that we have given to our story allows to establish a direct bridge with the algorithmic aspects of (reflexive) transitive closure, which is clearly not an alien topic to the training of a computer scientist. Our treatment is, admittedly, very sketchy, the unique excuse for this being the fact that the details are already present in many textbooks in algorithmics, for example \cite{AHU:1974}. 

Consider the boolean algebra $$\mathbf{B} = (\{ T,F \}, \vee, \wedge, \overline{\phantom{xx}}),$$ 
then consider a finite  set $X$ with a fixed enumeration of its elements, $X = \{x_1,\ldots,x_n \}$ and the structure of $n \times n$ matrices over $\mathbf{B}$, which corresponds bijectively with the set of binary relations over $X$. The basic idea to be explained is that relational composition coincides with matrix product, as can easily be seen by considering the standard logical formula defining relational product
\[
(R \cdot S)(x,y) \eqdef \exists u \in X (R(x,u) \& S(u,y))
\]
and interpreting the right-hand side, in the style of algebraic logic, as the usual formula for matrix product
\[
\bigvee_{u \in X} \quad R_{xu} \wedge S_{uy}
\]
observing that $\bigvee_{u \in X}$ is a finite disjunction.
Then we can adapt the notion of reflexive transitive closure  given in Definition~\ref{def:tc}(\ref{def:tc-power}) 
\[
\bigvee_{0 \leq i \leq \# X}  R^i
\]
and conclude that (a not very efficient version of) the Warshall algorithm 
for the transitive closure of Boolean matrices arises directly from translating the latter formula into three nested for-loops \cite{AHU:1974,carre:1979}.

\begin{digression}
From the point of view of the formation of a computer scientist, the above interpretation is a simple exercise in programming. From the point of view of foundations, the above remarks are an introduction to a style of algorithm design that helps factoring out an abstract algorithm that can be specialized to several well-known algorithms (e.\ g., shortest path). While this area of algorithmics has become a chapter of standard textbooks, we believe that the methodological basis of such an interpretation should be studied as providing an independent insight into the use of algebraic and logical structures in computer science. There are several ways this can be done, and we just point out some references that may help finding an orientation in the wide range of algebraic structures that may serve this purpose. One basic notion is that of a \emph{path algebra} \cite[Ch.\ 3]{carre:1979}, of which the Boolean algebra of truth values $\mathbf{B}$ is an example. Taking $n \times n$ Boolean matrices $X,Y$ with values in $\mathbf{B}$ one can define \cite[\S 3.2.5]{carre:1979} sum and product operations 
\[
(X + Y)_{ij} := X_{ij} \vee Y_{ij}
\]
\[
(X \cdot Y)_{ij} := \bigvee_{k = 1}^n (X_{ik} \wedge Y_{kj})
\]
which give the set of $n \times n$ matrices, again, the structure of path algebra. Often a stronger \emph{closed semiring} structure \cite[\S 5.6]{AHU:1974} is adopted. Quantales are also structures that allow the same kind of algebraic treatment. Going beyond the specific problem of calculating the reflexive transitive closure of a Boolean matrix, we can include a still wider range of algebraic structures, see the discussion in \cite[\S 2]{pratt:1991}.
\end{digression}

\section{Further topics}\label{further}
\newcommand{\finsubset}{\subseteq_{\mathrm{fin}}}
\newcommand{\corners}[1]{\langle #1 \rangle}
\newcommand{\set}[1]{\{ #1 \}}

We have presented a chapter of what might look very much like a traditional approach, but it is not.

While the topics that we have discussed might be part of a ``vertical'' approach that consists of a list of subsections of a syllabus of a course in foundations, our emphasis is actually on a complementary ``horizontal'' approach, trying to isolate the technical fundamentals (as opposed to foundations) that students should be able to learn from the hands-on work suggested by our sketch in the previous sections. Admittedly, our presentation is biased towards algebraic structure: we have tried to provide examples that might harmonize with the course in discrete mathematics that is assumed to be carried out in parallel with that on foundations. 
A set of examples that we would like to include in our lectures could be taken from relational algebra, abstract rewriting systems, and possibly also from the most abstract parts of classical recursion theory, for example Owing's treatment of diagonalization \cite[II.2, p.\ 154]{odifreddi:1989}. 

Finally, we suggest a theme that connects well to the general setting of our presentation.

\paragraph{Coinduction}
A natural complement to the development of transitive closure is an introduction to coinduction. Beside its central role in proving properties of structures that may unfold to infinity, like streams and labelled transitions systems, it offers a dual approach to the one we have followed. While, in defining transitive closure, we have been looking for the \emph{smallest} objects satisfying certain properties, in coinduction one looks for the \emph{largest}. In a set-theoretic presentation, coinduction (and the fundamental notion of bisimulation) require no more than the complete lattice of binary relations under inclusion, and this makes the subject a viable topic for the same kind of treatment that we have given to transitive closure. 

 \section*{Some literature}

\begin{itemize}
    \item 
Ancestral relations appear for the first time in~\cite{frege1879begriffsschrift}.
They are also discussed in Section E of \cite{russell1910pm}, which is about the inductive relations. See also \S39 in~\cite{quine1940ml}.

\item Some historical information on transitive closure, with a motivation from fundations, is contained in~\cite{pratt1992lics}. Vaughan Pratt also reformulated abstractly the Floyd-Warshall algorithm for transitive closure of relations (= Boolean matrices) in \cite{Pr89a}. The history of the Floyd-Warshall algorithm is in itself interesting: after the original publication in \cite {warshall:1962} a correctness proof has appeared in the classic \cite{ahoUllman1972}. This proof has been discussed and improved in \cite{wegner:1974} and further studied in \cite{ADJ:1976}.

\item The literature on Kleene algebras and related structures is extensive: the main reference is \cite{conway1971regular}. The name is due to Dexter Kozen: his page (\url{https://www.cs.cornell.edu/~kozen/}) contains links to lectures and many of Kozen's papers devoted to this topic. Although not directly focussed on Kleene algebra, \cite{pratt:1991} is a very useful introduction to the whole area.

\item The origins of coinduction have been studied in a nice paper by Sangiorgi \cite{sangiorgi:2009}. Coinduction (and its dual) are the topics of \cite{JacobsRutten2011}, The whole subject is covered in the two volumes \cite{Sangiorgi2011,sangiorgiRutten2011}. Applications of coalgebraic notions which may even be appreciated by freshmen are the subject of \cite{rutten:2003} and of several other papers by the same author.
\end{itemize}

\bibliographystyle{eptcs}
\bibliography{biblio}

\end{document}